\documentclass[envcountsame]{llncs}

\usepackage{latexsym,amsfonts}
\usepackage{amsmath}
\usepackage{amssymb}
\usepackage{array}

\usepackage{graphicx,color}
\usepackage{latexsym,amsfonts}
\usepackage{amsmath}
\usepackage{amssymb}
\usepackage{bbm}

\usepackage{algorithmic}
\usepackage{algorithm}
\usepackage{tikz}
\usepackage{url}
\usepackage{fullpage}

\colorlet{green}{green!50!black}
\usepackage[breaklinks=true]{hyperref} 
\hypersetup{
	colorlinks=true,
	citecolor=green,
	linkcolor=blue,  
}
\newcounter{lpnumber} 
\newenvironment{linearprogram}[1]
{ \stepcounter{lpnumber}
  \begin{gather} #1 \tag{LP\arabic{lpnumber}} \\[-5ex] \notag
  \end{gather}
  \hspace{1.5cm} subject to \\[-3ex]
  \align }
{ \endalign }
\newcommand{\minimize}[1]{\text{minimize} \ #1}
\newcommand{\maximize}[1]{\text{maximize} \ #1}

\newcommand{\cost}{\mathsf{cost}}
\newcommand{\wt}{\mathsf{wt}}

\newcommand{\Nbr}{\mathsf{Nbr}}

\newtheorem{new-claim}{Claim}

\begin{document}
\title{Maximum Matchings and Popularity}
\author{Telikepalli Kavitha}
\institute{Tata Institute of Fundamental Research, Mumbai, India\\\email{kavitha@tifr.res.in}}
\maketitle
\pagestyle{plain}

\begin{abstract}
Let $G = (A \cup B,E)$ be a bipartite graph where every node has a strict ranking of its neighbors. For every node, its preferences
over neighbors extend naturally to preferences over matchings.
Matching $N$ is ``more popular'' than matching $M$ if the number of nodes that prefer $N$ to $M$ is more than the number that prefer 
$M$ to $N$. 
A maximum matching $M$ in $G$ is a {\em popular max-matching} if there is no maximum matching in $G$ that is more 
popular than $M$. Such matchings are relevant in applications where the set of admissible solutions is the set of maximum matchings
and we wish to find a {\em best} maximum matching as per node preferences. 

It is known that a popular max-matching always exists in $G$. Here we show
a compact extended formulation for the popular max-matching polytope. So when there are edge costs, 
a min-cost popular max-matching can be computed in polynomial time. This is in contrast to the min-cost popular matching problem which is 
known to be NP-hard. We also consider {\em Pareto-optimality}, which is a relaxation of popularity, and show that computing a min-cost Pareto-optimal matching/max-matching is NP-hard.
\end{abstract}

\section{Introduction}
\label{sec:intro}
Consider a bipartite graph $G = (A \cup B, E)$ on $n$ nodes and $m$ edges where every node ranks its neighbors in a strict order of 
preference. The goal is to find an optimal matching in $G$ and the typical notion of optimality in such an instance is 
{\em stability}. A matching $M$ is stable if no edge blocks $M$; an edge $(a,b)$ blocks $M$ if $a$ and $b$ prefer 
each other to their respective assignments in $M$. Gale and Shapley~\cite{GS62} in 1962 showed that stable matchings always exist in
$G$ and gave a linear time algorithm to find one.

 Stable matchings and their variants are used in many real-world applications such as matching students to schools and colleges~\cite{AS03,BCCKP18}. Observe that stability empowers every edge with the power to block a matching -- thus this is a rather strict notion. Popularity 
is a natural relaxation of stability that captures collective decision. Preferences of a node over its neighbors extend naturally to 
preferences over matchings as follows.
 We say node $u$ prefers matching $M$ to matching $N$ if (i)~$u$ is matched in $M$ and unmatched in $N$ or (ii)~$u$ is matched in 
both $M$ and $N$ and it prefers its partner in $M$ to its partner in $N$.

 Let $\phi(M,N)$ be the number of nodes that prefer $M$ to $N$ and similarly, let $\phi(N,M)$ be the number of nodes that prefer $N$ 
to $M$. We say $N$ is more popular than $M$ if $\phi(N,M) > \phi(M,N)$. 

\begin{definition}
\label{def:popular}
A matching $M$ is popular if there is no matching in $G$ that is more popular than $M$, i.e.,  $\Delta(M,N) \ge 0$ for all matchings $N$ in $G$, where $\Delta(M,N) = \phi(M,N) - \phi(N,M)$.
\end{definition}

The notion of popularity was introduced by G\"ardenfors~\cite{Gar75} in 1975 and he observed that every stable matching is popular. One of the main merits of popularity is that it allows larger matchings compared to stable matchings. All stable matchings have the same size~\cite{GS85} which may be as low as $|M_{\max}|/2$ (where $M_{\max}$ is a maximum matching in $G$) and it is known that there is always a popular matching of size at least $2|M_{\max}|/3$ and there are simple instances with no larger popular matching~\cite{Kav14}. 

\paragraph{\bf Maximum Matchings.}
In the domain of {\em matchings under preferences}, there are applications where the size of the matching computed is of utmost importance,
say, in assigning doctors to hospitals during a pandemic. For instance, during the covid-19 pandemic in Mumbai, public hospitals were overwhelmed with rising number of patients and were severely short-staffed; so doctors associated with private clinics were asked to also work in public hospitals~\cite{indian-express,the-week}. This problem can be modeled as described below.\footnote{In the doctors-hospitals setting, one typically seeks a many-to-one matching since a hospital may need several doctors. Here we focus on one-to-one matchings.}

Let $A$ be the set of doctors and $B$ be the set of hospitals: edge $(a,b)$ denotes that doctor $a$ and hospital~$b$ are interested in each other. Every doctor has a strict ranking of possible hospitals and similarly, every hospital has a strict ranking of possible doctors. We want the maximum number of doctors to get assigned to hospitals, i.e., we do not wish to compromise at all on the size of our matching. So the set of admissible solutions is the set of maximum matchings and among all maximum matchings, we wish to find a {\em best} matching as per node preferences. In the context of popularity, this means no maximum matching is ``more popular'' than the maximum matching we find. So what we seek is a {\em popular max-matching}, which is defined below.

\begin{definition}
\label{def:pop-max-mat}
Call a maximum matching $M$ in $G = (A\cup B, E)$ a popular max-matching if $\Delta(M,N) \ge 0$ for all maximum matchings $N$ 
in $G$.
\end{definition}

Note that a popular max-matching $M$ need not be popular in the set of {\em all} matchings---so infeasible matchings 
(i.e., those that are not maximum) may be more popular than $M$, however $M$ is popular {\em within} the set of 
maximum matchings in $G$. The relation ``more popular than'' is not transitive and since stable matchings need not be maximum matchings,
it is not clear if every instance admits a popular max-matching. This question was considered in \cite{Kav14} where
it was shown that popular max-matchings always exist in $G = (A \cup B, E)$ and one such matching can be computed in $O(mn)$ time.

Suppose there is a cost function $c: E \rightarrow \mathbb{R}$ where 
$c(e)$ is the cost of including the edge $e$ in our matching; for any matching $M$,  $c(M) = \sum_{e \in M}c(e)$.
We consider the problem of computing a min-cost popular max-matching in $G$.
Solving the min-cost popular max-matching problem efficiently implies efficient algorithms for a whole host of  popular
max-matching problems such as finding one with forced/forbidden edges or one with max-utility or one with min-regret.
In general, a cost function allows us to ``access'' the entire set of popular max-matchings;
note that $G$ may have more than $2^n$ such matchings~\cite{Thu02}. 

While it is easy to find a min-size/max-size popular matching in $G$~\cite{GS62,HK11,Kav14}, it is NP-hard to compute a min-cost popular
matching~\cite{FKPZ18}. On the other hand, there are several polynomial time algorithms to compute a min-cost stable matching and some 
variants of this problem~\cite{Fed92,Fed94,Fle03,ILG87,Rot92,TS98,VV89}. Moreover, the stable matching polytope of $G$
has an elegant linear-size description in $\mathbb{R}^m$~\cite{Rot92}.

Let ${\cal M}_G$ denote the popular max-matching polytope of $G$. That is, the extreme points of the polytope 
${\cal M}_G$ are all and only the edge incidence vectors of popular max-matchings in $G$.
A compact description of ${\cal M}_G$ (or some extension\footnote{A polytope $Q$ that linearly projects to a polytope $P$ is an \emph{extension} of $P$ and a linear description of $Q$ is an \emph{extended formulation} for $P$. The minimum size of an extension of $P$ is the \emph{extension complexity} of $P$.} of it) implies a polynomial time algorithm to compute a min-cost popular max-matching. Our main theorem is 
that the polytope ${\cal M}_G$ has a compact extended formulation. So unlike the min-cost popular matching problem, interestingly and quite
surprisingly, the  min-cost popular max-matching problem is tractable.

\begin{theorem}
\label{thm:first}
Given $G = (A \cup B, E)$ where nodes have strict preferences and $c: E \rightarrow \mathbb{R}$, 
a min-cost popular max-matching can be found in polynomial time.
\end{theorem}

\noindent{\bf A relaxation of popularity.}
Pareto-optimality is a natural relaxation of popularity that any
reasonable matching in this domain should satisfy. If $M$ is a matching that is not Pareto-optimal then there is a
{\em better} matching $N$, i.e., no node is worse-off in $N$ than in $M$ and at least one node is better-off.
The {\em unpopularity factor} of $M$ is defined as follows~\cite{McC06}:
$u(M) = \max_{N\ne M}\phi(N,M)/\phi(M,N)$.

A matching $M$ is Pareto-optimal if $u(M) < \infty$. 
So there is no matching $N$ such that $\phi(N,M) > 0$ and $\phi(M,N) = 0$.
Observe that a popular matching $M$ satisfies $u(M) \le 1$.
A maximum matching $M$ that satisfies $u(M) < \infty$ is a
{\em Pareto-optimal max-matching}. 
We show that it is NP-hard to find a min-cost Pareto-optimal matching/max-matching (see Theorem~\ref{thm:pareto-opt}).

\paragraph{\bf Background.}
Many algorithmic questions in popular matchings have been investigated in the last 10-15 years;
we refer to~\cite{Cseh} for a survey. In the domain of popular matchings with two-sided preferences (every node has
a preference order ranking its neighbors), other than a handful of positive results~\cite{CK18,HK11,Kav14,Kav20},
most optimization problems (incl. finding a popular matching that is not a min-size/max-size popular matching)
have turned out to be NP-hard~\cite{FKPZ18}. 
Computing a min-cost {\em quasi-popular} matching $M$, i.e., $u(M)\le 2$, is also NP-hard~\cite{FK20}.

Compact extended formulations for the {\em dominant} matching\footnote{These are popular matchings that are more popular than all larger matchings.} polytope~\cite{CK18,FK20} and the popular fractional matching polytope~\cite{Kav16,KMN09} are  known but the popular matching polytope has near-exponential extension complexity~\cite{FK20}.
Thus in contrast to stable matchings, the landscape of popular matchings has only a few positive results. 
However, within the set of maximum matchings, the notion of popularity seems more suitable than stability for algorithmic results.
In the context of stability, a {\em best} maximum matching would be one with the least number of blocking edges. The problem of finding a
maximum matching with the least number of blocking edges was considered in \cite{BMM10} and was shown to be NP-hard. 

Though an $O(mn)$ time algorithm to find a popular max-matching in $G$ is known~\cite{Kav14}, there are no previous results on the {\em optimization} variant of this problem, i.e., to find a min-cost popular max-matching, which is a natural problem here.
It is common in this domain to have an efficient algorithm to find a max-size matching in some class -- say, Pareto-optimal matchings
for one-sided preferences (only nodes in $A$ have preferences) studied in  \cite{ACMM04} or popular matchings~\cite{HK11,Kav14}, however
finding a {\em min-cost} matching in such a class is NP-hard~\cite{ACMM04,FKPZ18}.
Thus the existence of an efficient algorithm to find some  popular max-matching was no guarantee on the tractability of the
min-cost popular max-matching problem.

\paragraph{\bf Our Techniques.}
An algorithm called the ``$|A|$-level Gale-Shapley algorithm'' was given in \cite{Kav14} to find a popular max-matching in $G = (A \cup B, E)$. As we show in Section~\ref{sec:char}, this algorithm is equivalent to running Gale-Shapley algorithm in an auxiliary instance $G^*$ with $|A|$ copies of each node in $A$ and the proof in \cite{Kav14} can be easily adapted to show that there is a linear map from the set of stable matchings in $G^*$ to the set of popular max-matchings in $G$ (see Theorem~\ref{thm:stable}). Our novel contribution here is to show that this map is {\em surjective}, i.e., every popular max-matching in $G$ is the image of a stable matching in $G^*$

Our proof that every popular max-matching in $G$ has a stable matching in $G^*$ as its preimage is based on LP-duality. Here we introduce the notion of dual certificates for popular max-matchings---dual certificates for popular matchings are well-understood: these are elements in $\{0,\pm 1\}^n$~\cite{Kav16}. We show in Section~\ref{sec:dual-certificate} that every popular max-matching in $G$ has a dual certificate $\vec{\alpha} \in \mathbb{Z}^n$ where $\alpha_a \in \{0,-2,-4,\ldots\}$ for $a \in A$ and $\alpha_b \in \{0,2,4,\ldots\}$ for $b \in B$
and $\alpha$-values for some nodes can be restricted to fixed values (see Theorem~\ref{thm:main}). 
Given a popular max-matching $M$ in $G$ and such a dual certificate $\vec{\alpha}$, we use $\vec{\alpha}$ to construct a preimage $S$ 
for $M$ and prove that $S$ is a stable matching in $G^*$.

\section{Popular Max-Matchings}
\label{sec:char}

In this section we first show a simple characterization of popular max-matchings. Then we show a method to construct matchings that satisfy this characterization.

Let $M$ be any matching in $G = (A \cup B, E)$. The following edge weight function $\wt_M$ will be useful here. For any $(a,b) \in E$:
\begin{equation*} 
\mathrm{let}\ \wt_M(a,b) = \begin{cases} 2   & \text{if\ $(a,b)$\ blocks\ $M$;}\\
	                     -2 &  \text{if\ $a$\ and\ $b$\ prefer\ their\ assignments\ in\ $M$\ to\ each\ other;}\\			
                              0 & \text{otherwise.}
\end{cases}
\end{equation*}

So $\wt_M(e) = 0$ for every $e \in M$.
For any cycle/path $\rho$ in $G$, let $\wt_M(\rho) = \sum_{e\in\rho}\wt_M(e)$. Theorem~\ref{thm:char} uses this edge weight function to 
characterize popular max-matchings.

\begin{theorem}
  \label{thm:char}
  For any maximum matching $M$ in $G$, $M$ is a popular max-matching if and only if 
  (1)~there is no alternating cycle $C$ with respect to $M$ such that $\wt_M(C) > 0$ and
  (2)~there is no alternating path $p$ with an unmatched node as an endpoint such that $\wt_M(p) > 0$.
\end{theorem}
\begin{proof}
Let $M$ be a popular max-matching. We need to show that conditions~(1) and (2) given in the theorem statement hold.
Suppose not. Then there exists either an alternating path with an unmatched node as an endpoint or an alternating cycle
wrt $M$ (call this path/cycle $\rho$) such that $\wt_M(\rho) > 0$.
Since $\wt_M(e) \in \{0, \pm 2\}$, $\wt_M(\rho) \ge 2$.
  
  Consider $N = M \oplus \rho$. This is a maximum matching in $G$ and observe that $\Delta(N,M) \ge \wt_M(\rho) -1$. We are subtracting
  $1$ here to count for that endpoint of $\rho$ (when $\rho$ is a path) that is matched in $M$ but will become unmatched in $N$. Since
  $\wt_M(\rho) \ge 2$, $\Delta(N,M) \ge 1$. So $N$ is more popular than $M$: this is a contradiction to $M$'s popularity within
  the set of maximum matchings. Thus conditions~(1) and (2) have to hold.

To show the converse, suppose $M$ is a maximum matching that obeys conditions~(1) and (2). Consider the symmetric difference
$M \oplus N$, where $N$ is any maximum matching in $G$ and let $C$ be any alternating cycle here. We know from (1) that $\wt_M(C) \le 0$.
Let $p$ be any alternating path in $M \oplus N$. Since $M$ and $N$ are maximum matchings, $p$ is an alternating path with
exactly one node not matched in $M$ as an endpoint. We know from (2) that $\wt_M(p) \le 0$. So we have:
  \[ \Delta(N,M) \ \le \ \sum_{\rho\in M \oplus N}\wt_M(\rho) \ \le\ 0.\]
Thus $M$ is popular within the set of maximum matchings in $G$. \qed 
\end{proof}  

\smallskip

\noindent{\bf An auxiliary instance $G^*$.}
We will now construct a new instance $G^* = (A^* \cup B^*, E^*)$ such that every stable matching in $G^*$ maps to a maximum matching in $G$
that satisfies properties~(1) and (2) given in Theorem~\ref{thm:char}. The structure of the instance $G^*$ is inspired by an instance used in \cite{CK18} whose
stable matchings map to dominant matchings in $G$. 

We describe the node sets $A^*$ and $B^*$. Let $n_0 = |A|$. 
For every $a \in A$, the set $A^*$ has $n_0$ copies of $a$: call them 
$a_0,\ldots,a_{n_0-1}$. So $A^* = \cup_{a\in A}\{a_0,\ldots,a_{n_0-1}\}$. 

\smallskip

Let $B^* = \cup_{a\in A}\{d_1(a),\ldots,d_{n_0-1}(a)\}\cup\{\tilde{b}: b \in B\}$, where
$\tilde{B} = \{\tilde{b}: b \in B\}$ is a copy of the set $B$.
So along with nodes in $\tilde{B}$, the set $B^*$ also contains $n_0-1$ nodes 
$d_1(a),\ldots,d_{n_0-1}(a)$ for each $a \in A$. These will be called {\em dummy} nodes.

\paragraph{\bf The edge set.}
For each $(a,b) \in E$, the edge set $E^*$ contains $n_0$ edges $(a_i,\tilde{b})$ for $0 \le i \le n_0-1$. For each $a \in A$
and $i \in \{1,\ldots,n_0-1\}$, $E^*$ also has $(a_{i-1},d_i(a))$ and $(a_i,d_i(a))$. 
The purpose of $d_1(a),\ldots,d_{n_0-1}(a)$ is to ensure that in any stable matching 
in $G^*$, at most one node among $a_0,\ldots,a_{n_0-1}$ is matched to a neighbor in $\tilde{B}$.
For each $i \in \{1,\ldots,n_0-1\}$, the preference order of $d_i(a)$ is $a_{i-1} \succ a_i$.

\paragraph{\bf Preference orders.}
Let $a$'s preference order in $G$ be $b_1 \succ \cdots \succ b_k$. Then $a_0$'s preference order in $G^*$ is $\tilde{b}_1 \succ \cdots \succ \tilde{b}_k \succ d_1(a)$, i.e., it is analogous to $a$'s preference order in $G$ with $d_1(a)$ added as $a_0$'s last choice.

\begin{itemize}
\item  For $i \in \{1,\ldots,n_0-2\}$, the preference order of $a_i$ in $G^*$ is as follows: 
$d_i(a) \succ \tilde{b}_1 \succ \cdots \succ \tilde{b}_k \succ d_{i+1}(a)$. So $a_i$'s top choice is $d_i(a)$ and last choice is
$d_{i+1}(a)$.

\item The preference order of $a_{n_0-1}$ is $d_{n_0-1}(a) \succ \tilde{b}_1 \succ \cdots \succ \tilde{b}_k$.
\end{itemize}

Since each of $a_0,\ldots,a_{n_0-2}$ and $d_1(a),\ldots,d_{n_0-1}(a)$ is a top choice neighbor for some node, every stable matching in $G^*$ 
has to match all these nodes. So the only node among 
$a_0,\ldots,a_{n_0-1},d_1(a),\ldots,d_{n_0-1}(a)$ that can possibly be left unmatched in a stable matching in $G^*$ is $a_{n_0-1}$.

Consider any $b \in B$ and let its preference order in $G$ be $a \succ \cdots \succ z$. Then the preference order of $\tilde{b}$
in $G^*$ is
\[\underbrace{a_{n_0-1}\succ\cdots \succ z_{n_0-1}}_{\text{all subscript $n_0-1$ neighbors}}\succ\underbrace{a_{n_0-2} \succ \cdots \succ z_{n_0-2}}_{\text{all subscript $n_0-2$ neighbors}}\succ\cdots\succ\underbrace{a_0\succ\cdots \succ z_0}_{\text{all subscript 0 neighbors}} \]

That is, $\tilde{b}$'s preference order in $G^*$ is all its subscript~$n_0-1$ neighbors, followed by all its subscript~$n_0-2$
neighbors, so on, and finally, all its subscript~0 neighbors. For each $i \in \{0,\ldots,n_0-1\}$: within all subscript~$i$ neighbors,
the order of preference for $\tilde{b}$ in $G^*$ is $b$'s order of preference in $G$.

\paragraph{\bf The set $S'$.} For any stable matching $S$ in $G^*$, define $S' \subseteq E$ to be the set of edges
obtained by deleting edges in $S$ that are incident to dummy nodes 
and replacing any edge $(a_i,\tilde{b}) \in S$ with the original edge $(a,b) \in E$.
Since $S$ matches at most one node among $a_0,\ldots,a_{n_0-1}$ to a neighbor in $\tilde{B}$, the set $S'$ is a valid matching in $G$.
The proof of Theorem~\ref{thm:stable} is based on the proof of correctness of the $|A|$-level Gale-Shapley algorithm (from \cite{Kav14})
in the original instance $G$. 

\begin{theorem}
  \label{thm:stable}
  Let $S$ be a stable matching in $G^*$. Then $S'$ is a popular max-matching in $G$.
\end{theorem}
\begin{proof}
  Partition the set $A$ into $A_0 \cup \cdots \cup A_{n_0-1}$ where for $0 \le i \le n_0-2$:
  $A_i = \{a \in A: (a_i,\tilde{b}) \in S$ for some $\tilde{b} \in \tilde{B}\}$, i.e., $a_i$ is matched in $S$ to a
  neighbor in $\tilde{B}$. The left-out nodes in $A$, i.e., those in $A \setminus (A_0\cup\cdots\cup A_{n_0-2})$, 
  form the set $A_{n_0-1}$ (see Fig.~\ref{fig:levels}).

  Similarly, partition $B$ into $B_0 \cup \cdots \cup B_{n_0-1}$ where for $1 \le i \le n_0-1$: $B_i = \{b: (a_i,\tilde{b}) \in S$ for some 
  $a \in A_i\}$, i.e., $\tilde{b}$'s partner in $S$ is a subscript~$i$ node. Let $B_0 = B \setminus (B_1\cup\cdots\cup B_{n_0-1})$
  be the set of left-out nodes in $B$. 
  \begin{figure}[h]
  \centerline{\resizebox{0.38\textwidth}{!}{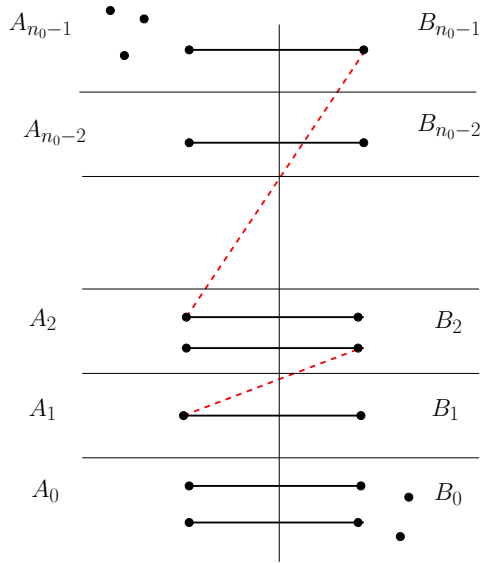}}
  \caption{$A = A_0 \cup \cdots \cup A_{n_0-1}$ and $B = B_0 \cup \cdots \cup B_{n_0-1}$. 
    All nodes unmatched in $S'$ are in $A_{n_0-1} \cup B_0$. The dashed edges are blocking edges to $S'$.}
  \label{fig:levels}
  \end{figure}

  The following properties hold: (these are proved in the appendix)
  \begin{enumerate}
  \item $S' \subseteq \cup_{i=0}^{n_0-1} (A_i\times B_i)$. Moreover, $S'$ restricted to each set $A_i \cup B_i$ is stable.
  \item For any $i$ and edge $(a,b)$ where $a \in A_{i+1}, b \in B_i$: we have $\wt_{S'}(a,b) = -2$.
  \item $G$ has no edge in $A_i \times B_j$ where $i \ge j+2$.
  \item Any blocking edge to $S'$ has to be in $A_i \times B_j$ where $i \le j-1$.
  \item All nodes that are unmatched in $S'$ are in $A_{n_0-1} \cup B_0$.
  \item $S'$ is a maximum matching in $G$.
  \end{enumerate}
  
  Properties~1-4 imply that for any alternating cycle $C$ wrt $S'$, we have $\wt_{S'}(C) \le 0$. 
  Properties~1-5 imply that for any alternating path $p$ with one unmatched
  node as an endpoint, we have $\wt_{S'}(p) \le 0$. We refer to \cite[Theorem~2]{Kav14} for more details.
  Property~6 states that $S'$ is a maximum matching in $G$.
  Hence $S'$ is a popular max-matching in $G$ (by Theorem~\ref{thm:char}). \qed
\end{proof}    

Thus every stable matching in $G^*$ maps to a popular max-matching in $G$. We will show in Section~\ref{sec:dual-certificate}
that {\em every} popular max-matching in $G$ has to be realized in this manner. This is the tough part and
we will use dual certificates here. We will see that dual certificates capture a useful profile of popular max-matchings.

\section{Dual Certificates}
\label{sec:dual-certificate}

We first consider the simple case when $G$ admits a perfect matching. 
Let $M$ be a popular perfect matching in $G$. So no perfect matching is more popular than $M$. Consider the following linear program 
\eqref{LP1} that computes a max-weight (with respect to $\wt_M$) perfect matching in $G$.

\begin{linearprogram}
  {
    \label{LP1}
    \maximize \sum_{e \in E} \wt_M(e)\cdot x_e  
  }
  \qquad\sum_{e \in {\delta}(u)}x_e \ & = \ \ 1  \ \ \, \forall\, u \in A \cup B \ \ \ \ \ \text{and} \ \ \ \ \ x_e  \ \ge \ \ 0   \ \ \ \forall\, e \in E. \notag
\end{linearprogram}

It follows from the definition of the function $\wt_M$ that the optimal value of \eqref{LP1} is $\max_N \Delta(N,M)$ where $N$ is a
perfect matching in $G$. 
So if $M$ is a popular perfect matching then the optimal value of \eqref{LP1} is 0, which is $\Delta(M,M)$, i.e.,
the edge incidence vector of $M$ is an optimal solution to \eqref{LP1}.

The linear program \eqref{LP2} is the dual of \eqref{LP1}. Hence if $M$ is a popular perfect matching then there
exists a dual feasible $\vec{\alpha}$ such that $\sum_{u\in A \cup B} \alpha_u = 0$.

\begin{linearprogram}
  {
    \label{LP2}
    \minimize \sum_{u \in A \cup B}y_u
  }
  y_{a} + y_{b} \ & \ge \ \ \wt_M(a,b) \ \ \ \ \forall\, (a,b)\in E. \notag
\end{linearprogram}

Let $\vec{\alpha}$ be an optimal dual solution.
Observe that there exists an integral optimal solution to \eqref{LP2} since the constraint matrix is totally unimodular. Thus we can
assume that $\vec{\alpha} \in \mathbb{Z}^{2n_0}$, where $n_0 = |A| = |B|$. 

\begin{lemma}
\label{lemma:alpha}
If $M$ is a popular perfect matching in $G$ then there exists an optimal solution $\vec{\alpha}$ to \eqref{LP2}  such that
$\alpha_a \in \{0, -2, -4,\ldots, -2(n_0-1)\}$ for all $a \in A$ and
$\alpha_b \in \{0, 2, 4,\ldots, 2(n_0-1)\}$ for all $b \in B$.
\end{lemma}
\begin{proof}
The dual feasibility constraints are $\alpha_a + \alpha_b \ge \wt_M(a,b)$ for all $(a,b) \in E$.
For each edge $(a,b) \in M$: $\alpha_{a} + \alpha_{b} = \wt_M(a,b) = 0$ by complementary slackness. 
Since $\alpha_b = -\alpha_a$ for $(a,b) \in M$ and because $\wt_M(e) \in \{0,\pm 2\}$ for each edge $e$,
we can assume that in the sorted order of distinct $\alpha$-values taken by nodes in $A$,
for any two consecutive values $\alpha_{a'},\alpha_{a''}$, where $\alpha_{a'} > \alpha_{a''}$, we have $\alpha_{a'} - \alpha_{a''} = 2$. 
Thus we can assume that $\alpha_a \in \{k,k-2, k-4, \ldots, k-2(n_0-1)\}$ 
for all $a \in A$ and $\alpha_b \in \{-k,-k+2,-k+4,\ldots, -k+2(n_0-1)\}$ for all $b \in B$, for some $k \in \mathbb{Z}$. 

Observe that $k$ has no impact on the objective function since $|A| = |B|$: so $k$'s and $-k$'s cancel each other out.
Let us update $\vec{\alpha}$ as follows: $\alpha_a = \alpha_a - k$ for every $a \in A$ and
 $\alpha_b = \alpha_b + k$ for every $b \in B$.
The updated vector $\vec{\alpha}$ continues to be dual feasible since $\alpha_a + \alpha_b$, for any edge $(a,b)$, is unchanged
by this update. Thus there is an optimal solution $\vec{\alpha}$ to \eqref{LP2} such that
$\alpha_a \in \{0,-2, -4, \ldots, -2(n_0-1)\}$ for all $a \in A$ and
$\alpha_b \in \{0,2, 4, \ldots, 2(n_0-1)\}$ for all $b \in B$. \qed
\end{proof}

Let $M$ be a popular perfect matching in $G$.
In order to define a stable matching $S$ in $G^*$ such that $M = S'$ (the set $S'$ is defined above Theorem~\ref{thm:stable}), 
we will use the vector $\vec{\alpha}$ described in Lemma~\ref{lemma:alpha}. Since $M$ is perfect, we know that for any $a \in A$, 
there is an edge $(a,b) \in M$ for some neighbor
$b$ of $a$. Recall that $\alpha_a + \alpha_b = \wt_M(a,b) = 0$ by complementary slackness.
We will include the edge $(a_i,\tilde{b})$ in $S$ where $\alpha_a = -2i$ and $\alpha_b = 2i$. Thus we define $S$ as follows:
\[S = \cup_{i=0}^{n_0-1}\{(a_i,\tilde{b}): (a,b) \in M \ \text{and}\ \alpha_a = -2i, \alpha_b = 2i\} \cup \{\text{necessary edges incident to 
dummy nodes in}\ G^*\}.\]

In more detail, the edges incident to dummy nodes that are present in $S$ are as follows: if 
$\alpha_a = -2i$ then these edges are $(a_j,d_{j+1}(a))$ for $0 \le j \le i-1$ and $(a_j,d_{j}(a))$ for $i+1 \le j \le n_0-1$.

Since $(a_i,\tilde{b}) \in S$, all the $n_0$ nodes $a_0,\ldots,a_{n_0-1}$ and the
dummy nodes $d_1(a),\ldots,d_{n_0-1}(a)$ corresponding to $a$ in $G^*$ are matched in $S$. This holds for every $a \in A$. Also every
$\tilde{b} \in \tilde{B}$ is matched in $S$ since $M$ is a perfect matching in $G$. Thus $S$ is a perfect matching in $G^*$. 
It is easy to check that $S' = M$. What we need to prove is the stability of $S$ in $G^*$.

\begin{theorem}
  \label{thm:conv1}
  The matching $S$ is stable in $G^*$.
\end{theorem}
\begin{proof}
  We need to show there is no edge in $G^*$ that blocks $S$. There is no blocking edge incident to a dummy node:
  this is because a dummy node $d_i(a)$ has only two neighbors and when $d_i(a)$ is matched in $S$ to its second choice neighbor $a_i$,
  its top choice neighbor $a_{i-1}$ prefers its partner in $S$ to $d_i(a)$.
  
  Let us now show that no node in
  $a_0,\ldots,a_{n_0-1}$ has a blocking edge incident to it, for any $a \in A$.
  Let $(a_i,\tilde{b}) \in S$ where $(a,b) \in M$. All of $a_{i+1},\ldots,a_{n_0-1}$ are matched to their respective top choice neighbors
  $d_{i+1}(a),\ldots,d_{n_0-1}(a)$. So there is no blocking edge incident to  any of $a_{i+1},\ldots,a_{n_0-1}$.
  
  All of $a_0,\ldots,a_{i-1}$ are matched to their last choice neighbors---these are the dummy nodes $d_1(a),\ldots,$ $d_i(a)$, 
  respectively.
  Consider any neighbor $w \in B$ of $a$. We need to show that $\tilde{w} \in \tilde{B}$ is matched in $S$ to a neighbor preferred to 
  all of $a_0,\ldots,a_{i-1}$.  We have $\alpha_a + \alpha_w \ge \wt_M(a,w)$. Since $\alpha_a = -2i$ and $\wt_M(e) \ge -2$ for every edge 
  $e$, it follows that $\alpha_w \ge 2i - 2$.

  So $(z,w) \in M$ for some neighbor $z$ of $w$ such that $\alpha_z = -\alpha_w \le -(2i-2)$. Equivalently, $(z_j,\tilde{w}) \in S$ where 
  $j \ge i-1$. Thus there is no blocking edge between $\tilde{w}$ and any of $a_0,\ldots,a_{i-2}$ by $\tilde{w}$'s preference order in 
  $G^*$. We will now show that $(a_{i-1},\tilde{w})$ cannot be a blocking edge.
  \begin{itemize}
  \item  If $j \ge i$ then by $\tilde{w}$'s preference order in $G^*$, $\tilde{w}$ prefers $z_j$ to $a_{i-1}$ and so $(a_{i-1},\tilde{w})$ 
         does not block $S$.
  \item  If $j = i-1$ then $\wt_M(a,w) \le \alpha_a + \alpha_w = -2i + 2i - 2 = -2$. So both $a$ and $w$ prefer their respective
  partners in $M$ to each other. Thus $\tilde{w}$ prefers $z_{i-1}$ to $a_{i-1}$. So $(a_{i-1},\tilde{w})$ does not block $S$.
  \end{itemize}
  
  Finally, we need to show that there is no blocking edge incident to $a_i$. By the above argument, we only need to consider edges 
  $(a_i,\tilde{w})$ where $(z_i,\tilde{w}) \in S$. So $\wt_M(a,w) \le \alpha_a + \alpha_w = -2i + 2i = 0$. 
  Hence either $(a,w)\in M$ or one of $a,w$ prefers its partner in $M$ to the other. So
  either $(a_i,\tilde{w}) \in S$ or one of $a_i,\tilde{w}$ prefers its partner in $S$ to the other;
  thus the edge $(a_i,\tilde{w})$ does not block $S$. Hence no edge in $G^*$ blocks $S$. \qed
\end{proof}

\subsection{The general case}
We showed that when $G$ has a perfect matching, our mapping from the set of stable matchings in $G^*$ to the set of
popular max-matchings in $G$ is surjective. Now we look at the general case, i.e., $G$ need not have a perfect matching.

Let $M$ be a popular max-matching in $G$. Let $U \subseteq A \cup B$ be the set of nodes left unmatched in $M$.  
Consider \eqref{LP3} that computes a max-weight perfect matching  (with respect to $\wt_M$) 
in the subgraph $G'$ induced on $(A\cup B)\setminus U$. 

Let $E'$ be the edge set of $G'$. For any $v \in (A \cup B)\setminus U$, let $\delta'(v) = \delta(v) \cap E'$.
\begin{linearprogram}
  {
    \label{LP3}
    \maximize \sum_{e \in E'} \wt_M(e)\cdot x_e  
  }
  \qquad\sum_{e \in \delta'(v)}x_e \ & = \ \ 1  \ \ \, \forall\, v \in (A \cup B)\setminus U \ \ \ \ \ \text{and}\ \ \ \ \ x_e  \ \ge \ \ 0   \ \ \ \forall\, e \in E'. \notag
\end{linearprogram}

The optimal value of \eqref{LP3} is $\max_N\Delta(N,M)$ where $N$ is a perfect matching in $G'$. Any perfect matching in $G'$ is
a maximum matching in $G$ and since $M$ is a popular max-matching in $G$,
$\Delta(N,M) \le 0$ for any perfect matching $N$ in $G'$. 
Since $\Delta(M,M) = 0$, the edge incidence vector of $M$ is an optimal solution to \eqref{LP3}.
The linear program \eqref{LP4} is the dual of \eqref{LP3}.

\begin{linearprogram}
  {
    \label{LP4}
    \minimize \sum_{u \in (A \cup B)\setminus U}y_u
  }
  y_{a} + y_{b} \ & \ge \ \ \wt_M(a,b) \ \ \ \ \forall\, (a,b)\in E'. \notag
\end{linearprogram}

Let $n' = |(A\cup B)\setminus U|$. There exists an optimal solution $\vec{\alpha} \in \mathbb{Z}^{n'}$
to \eqref{LP4} such that $\sum_{u \in (A \cup B)\setminus U}\alpha_u = 0$.
Moreover, we can assume the following (see Lemma~\ref{lemma:alpha}) where $A' = A\setminus U$, $B' = B\setminus U$,
and $|A'| = |B'| = n_0$:

\begin{itemize}
\item[(1)]$\alpha_a \in \{0, -2, -4,\ldots,-2(n_0-1)\}$ for all $a \in A'$
\item[(2)] $\alpha_b \in \{0, 2, 4,\ldots, 2(n_0-1)\}$ for all $b \in B'$.
\end{itemize}

For any $T \subseteq A \cup B$, let $\Nbr(T)$ be the set of neighbors in $G$ of nodes in $T$.
Theorem~\ref{thm:main} is our main technical result here. 
Let $U_A = U \cap A$ and $U_B = U \cap B$.

\begin{theorem}
  \label{thm:main}
Let $M$ be a popular max-matching in $G$ and let $U$ be the set of nodes left unmatched in $M$.
  There exists an optimal solution $\vec{\alpha}$ to \eqref{LP4}, where $\alpha_a \in \{0, -2,\ldots,$ $-2(n_0-1)\}$
  for $a \in A', \alpha_b \in \{0, 2,\ldots, 2(n_0-1)\}$ for $b \in B'$ such that
  (i)~$\alpha_a = 0$ for $a \in \Nbr(U_B)$ and (ii)~$\alpha_b = 2(n_0-1)$ for $b \in \Nbr(U_A)$.
\end{theorem}
\begin{proof}
  We know there exists an optimal solution $\vec{\alpha}$ to \eqref{LP4} such that
  $\alpha_a \in \{0, -2,\ldots,-2(n_0-1)\}$ for $a \in A'$ and $\alpha_b \in \{0, 2,\ldots, 2(n_0-1)\}$ for $b \in B'$.
  Now we will update $\vec{\alpha}$ so that it remains
  an optimal solution to \eqref{LP4} in the above format and it also satisfies properties~(i) and (ii) given in the theorem
  statement.

 We will first update $\vec{\alpha}$ so that $\alpha_b = 2(n_0-1)$ for all $b \in \Nbr(U_A)$, i.e., property~(ii) is obeyed.
  We will then update $\vec{\alpha}$ so that it satisfies $\alpha_a = 0$ for all $a \in \Nbr(U_B)$, i.e.,
  property~(i) is obeyed. We then use the fact that $M$ is a maximum matching to show that
  the second update did not undo the ``good'' caused by the first update, i.e.,
  $\vec{\alpha}$ satisfies property~(i) and property~(ii).

  \smallskip
  
  \noindent{\em Property~(ii).}
  Suppose the original vector $\vec{\alpha} \in \{0,\ldots,\pm 2(n_0-1)\}^{n'}$ does not satisfy property~(ii).
  Then we will update $\vec{\alpha}$ so that property~(ii) is satisfied. First, we increase the $\alpha$-values
  of nodes in $\Nbr(U_A)$ to $2(n_0-1)$ and decrease the $\alpha$-values of their partners in $M$ to $-2(n_0-1)$.
  Now $\vec{\alpha}$ may no longer be a feasible solution to \eqref{LP4}. We will use the following three update
  rules for all $a \in A'$ to make $\vec{\alpha}$ feasible again. Let $\alpha_a = -2i$ where $i \in \{0,\ldots,n_0-1\}$.
  Suppose there is some $(a,b) \in E'$ such that $\alpha_a + \alpha_b < \wt_M(a,b)$. Let $M(b)$ be $b$'s partner in $M$.

  \begin{itemize}
  \item {\em Rule~1.} If $\wt_M(a,b) = 0$ then update $\alpha_b = 2i$ and $\alpha_{M(b)} = -2i$. 
  \item {\em Rule~2.} If $\wt_M(a,b) = -2$ then update $\alpha_b = 2(i-1)$ and $\alpha_{M(b)} = -2(i-1)$. 
  \item {\em Rule~3.} If $\wt_M(a,b) = 2$ then update $\alpha_b = 2(i+1)$ and $\alpha_{M(b)} = -2(i+1)$. 
  \end{itemize}
      
  At the onset, $\vec{\alpha}$ was a feasible solution to \eqref{LP4}: so $\alpha_a + \alpha_b \ge \wt_M(a,b)$ for $(a,b) \in E'$.
  Then we moved nodes in $\Nbr(U_A)$ and their partners in $M$ to sets $B_{n_0-1}$ and $A_{n_0-1}$, resp., where
  $A_i = \{a\in A': \alpha_a = -2i\}$ and $B_i = \{b\in B': \alpha_b = 2i\}$ for all $i$.
  The subscript $i$ will be called the {\em level} of nodes in $A_i \cup B_i$.
  
  The nodes that moved to $A_{n_0-1}$ have a lower $\alpha$-value than earlier and it is these
  nodes that ``pull'' their neighbors upwards to higher levels as given by rules~1-3.
  Let $a$ be a new node in level~$i$ and let $b$ be a neighbor of $a$ such that $\alpha_a + \alpha_b < \wt_M(a,b)$.
  Then $b$ and $M(b)$ move to:
  (1)~level~$i$ if $\wt_M(a,b) = 0$,
  (2)~level~$i-1$ if $\wt_M(a,b) = -2$, else
  (3)~level~$i+1$, i.e., if $\wt_M(a,b) = 2$.
  
  In turn, the nodes in $A'$ that have moved to these higher levels by rules~1-3 pull their neighbors and the partners of these
  neighbors upwards to higher levels by these rules. Thus we may get further new nodes in $B_{n_0-1},A_{n_0-1}$ and so on.
  While any of rules~1-3 is applicable, we apply that rule.
  So a rule may be applied many times to the same edge in $E'$.
  
   Claim~\ref{claim2} (proved in the appendix) shows a useful property. We show in its proof that such a blocking edge
   creates a {\em forbidden} alternating cycle/path wrt $M$ (as given in Theorem~\ref{thm:char}).

  \begin{new-claim}
    \label{claim2}
    By applying the above rules, suppose a node $v_0 \in A'$ moves to $A_{n_0-1}$. Then there is no blocking edge
    ($e$ such that $\wt_M(e) = 2$) incident to $v_0$.
  \end{new-claim}  
 
  Applying rules~1-3 increases the $\alpha$-values of some nodes in $B'$ and it never decreases the $\alpha$-value of any node
  in $B'$. The nodes in $B'$ with increased $\alpha$-values and their partners have moved to higher levels (see Fig.~\ref{fig:levels}).
  This upwards movement of nodes has to terminate at level~$n_0-1$. For the $\alpha$-value of any $b \in B'$ to be increased beyond
  $2(n_0-1)$, we need a blocking edge $(a,b)$ where $a \in A_{n_0-1}$: this would cause rule~3 to be applied which would increase
  $\alpha_b$ to $2n_0$. However there is no such blocking edge (by Claim~\ref{claim2}).
  
  Since there are $n_0$ levels and because $|B'| = n_0$, there can be at most $n_0^2$ applications of these rules.
  When no rule is applicable, $\vec{\alpha}$ is a feasible solution to \eqref{LP4}. Moreover, $\sum_u \alpha_u$
  is invariant
  under this update of $\alpha$-values, since we maintain $\alpha_a + \alpha_b = 0$ for every $(a,b) \in M$.
  Hence $\vec{\alpha}$ is an optimal solution to \eqref{LP4}. Thus for every popular max-matching $M$, there is an
  optimal solution $\vec{\alpha}$ to \eqref{LP4} in the desired format that satisfies property~(ii).
  
  \smallskip
  
 \noindent{\em Property~(i).}
  We now have an optimal solution $\vec{\alpha} \in \{0, \pm 2,\ldots,$ $\pm 2(n_0-1)\}^{n'}$ to \eqref{LP4}, where
  $\alpha_a \le 0$ for all $a \in A'$ and $\alpha_b \ge 0$ for all $b \in B'$, such that $\alpha_b = 2(n_0-1)$ for all $b \in \Nbr(U_A)$.
  Suppose property~(i) is not satisfied.
  
  Then we increase $\alpha$-values of certain nodes in $A'$ -- this moves several nodes {\em downwards}  with respect to
  their level (see Fig.~\ref{fig:levels}) and ensures that property~(i) holds.
  First, we increase the $\alpha$-values of nodes in $\Nbr(U_B)$ to 0 and their partners also have $\alpha$-values updated to $0$.
  Then we will use the
  following three update rules for all $b \in B'$. Let $\alpha_b = 2i$ where $i \in \{0,\ldots,n_0-1\}$.
  Suppose there is an edge $(a,b) \in E'$ such that $\alpha_a + \alpha_b < \wt_M(a,b)$.  Let $M(a)$ be $a$'s partner in $M$.
  
  \begin{itemize}
  \item {\em Rule~4.} If $\wt_M(a,b) = 0$ then update $\alpha_a = -2i$ and $\alpha_{M(a)} = 2i$. 
  \item {\em Rule~5.} If $\wt_M(a,b) = -2$ then update $\alpha_a = -2(i+1)$ and $\alpha_{M(a)} = 2(i+1)$. 
  \item {\em Rule~6.} If $\wt_M(a,b) = 2$  then update $\alpha_a = -2(i-1)$ and $\alpha_{M(a)} = 2(i-1)$. 
  \end{itemize}

 Applying rules~4-6 increases the $\alpha$-values of some nodes in $A'$ and it never decreases the $\alpha$-value of any node 
 in $A'$. The nodes in $A'$ with increased $\alpha$-values and their partners have moved to lower levels.
 Moreover, the movement of nodes downwards has to stop at level~$0$ since no blocking edge can be incident to any node that moves
 to $B_0$ (analogous to Claim~\ref{claim2}).
 While any of the above three rules is applicable, we apply that rule.
 When no rule is applicable, $\vec{\alpha}$ is a feasible solution to \eqref{LP4}.
 Since $\sum_u \alpha_u = 0$, $\vec{\alpha}$ is an optimal solution to \eqref{LP4}. So there is 
 an optimal solution $\vec{\alpha}$ to \eqref{LP4} in the desired format that satisfies property~(i).

 \smallskip

 \noindent{\em Properties~(i) and (ii).}
 Note that we cannot claim straightaway that the above $\vec{\alpha}$ satisfies both property~(i) and property~(ii).
 This is because rules~4-6 pull nodes downwards and this may have caused $\alpha_b < 2(n_0-1)$ for some $b \in \Nbr(U_A)$.
 We will now show this is not possible.
 Suppose there is such a node $w_1 \in \Nbr(U_A)$.

 So there is some $z\in U_A$ such that $(z,w_1) \in E$ and
 though $\alpha_{w_1} = 2(n_0-1)$ just before we started applying rules~4-6, the application of these rules caused $w_1$ to move 
 downwards, i.e., $\alpha_{w_1} < 2(n_0-1)$ at the end. 
 Initially we added nodes in $\Nbr(U_B)$ and their partners in $M$ to $A_0$ and $B_0$, respectively.
 Then we repeatedly applied rules~4-6 and this resulted in
 $w_1$ moving downwards from level~$n_0-1$. Corresponding to what caused $w_1$ to be ``pulled'' downwards, we will construct
 an alternating path $p = w_1 - M(w_1) - w_2 - M(w_2) - \cdots - w_r - M(w_r) - u$  between $w_1$ and some node $u \in U_B$.

 The path $p$ will consist of $r$ pairs of edges
 for some $r \ge 1$. For $1 \le i \le r-1$: the $i$-th pair of edges is $(w_i,M(w_i))$ and $(M(w_i),w_{i+1})$ where $w_{i+1}$ is the
 node that pulled $M(w_i)$ and $w_i$ downwards to their level due to the application of one of rules~4-6. The last pair of edges in $p$ is
 $(w_r,M(w_r))$ and $(M(w_r),u)$, where $w_r \in B_0, M(w_r) \in A_0$, and $u \in U_B$. Thus we have an alternating path $p$ wrt
$M$ between $w_1$ and $u \in U_B$.

By adding the edge $(z,w_1)$ as a prefix to the path $p$, we get an augmenting path $z - w_1 - \cdots - M(w_r) - u$ 
with respect to $M$. However there cannot be any augmenting path wrt $M$ since $M$ is a maximum matching in $G$.
Thus $\vec{\alpha}$ satisfies property~(ii). So we have an optimal solution $\vec{\alpha}$ to \eqref{LP4} in the desired format that satisfies
both property~(i) and property~(ii).  \qed
\end{proof}

\begin{theorem}
  \label{lem:surjective}
  Let $M$ be any popular max-matching in $G$. Then
  there is a stable matching $S$ in the graph $G^*$ such that $M = S'$.
\end{theorem}
\begin{proof}
Let $U$ be the set of nodes left unmatched in $M$.
A useful observation is that every node $v$ in $\Nbr(U)$ has to be matched in $M$ to some neighbor
that $v$ prefers to all its neighbors in $U$. Otherwise by replacing $(v,w) \in M$ with
$(v,u)$, where $v$ prefers $u \in U$ to $w$, we get a maximum matching more popular than $M$.

In order to construct the matching $S$ in $G^*$, we will use the vector $\vec{\alpha}$ described in Theorem~\ref{thm:main}.
Let $S = \cup_{i=0}^{n_0-1}\{(a_i,\tilde{b}): (a,b) \in M \ \text{and}\ \alpha_a = -2i, \alpha_b = 2i\} \cup \{$necessary edges incident to dummy nodes in $G^*\}$ where for each $a \in A'$: if 
$\alpha_a = -2i$ then the necessary edges incident to dummy nodes corresponding to $a$
are $(a_j,d_{j+1}(a))$ for $0 \le j \le i-1$ and $(a_j,d_{j}(a))$ for $i+1 \le j \le n_0-1$.
It is easy to check that $S' = M$.

We will now show that $S$ is a stable matching in $G^*$.
Consider any edge $(a_j,\tilde{w})$ in $G^*$ where $(a,w) \in E'$ and $0 \le j \le n_0-1$. The proof of Theorem~\ref{thm:conv1} shows
that $(a_j,\tilde{w})$ does not block $S$. 
Consider any edge $(a,w)$ in $E \setminus E'$.

\begin{itemize}
\item Suppose $a \in U_A$.
So $w \in \Nbr(U_A)$: hence $\alpha_w = 2(n_0-1)$ by property~(ii). So $(z_{n_0-1},\tilde{w}) \in S$
for some neighbor $z$ that $w$ prefers to $a$. Thus $(a_{n_0-1},\tilde{w})$ does not block $S$.
For $i \in\{0,\ldots,n_0-2\}$, none of the edges $(a_i,\tilde{w})$ can block~$S$ (by $\tilde{w}$'s preference order).

\item Suppose $w \in U_ B$. So $a \in \Nbr(U_B)$: hence $\alpha_a = 0$ by property~(i).
Thus $(a_0,\tilde{b}) \in S$ for some neighbor $b$ that $a$ prefers to $w$. 
Hence  $(a_0,\tilde{w})$ does not block $S$. Moreover, none of the edges $(a_i,\tilde{w})$ for $i\in\{1,\ldots,n_0-1\}$ can
block $S$ since $a_1,\ldots,a_{n_0-1}$ are matched to their respective top choice neighbors $d_1(a),\ldots,d_{n_0-1}(a)$.
\end{itemize}
Finally, no edge incident to a dummy node blocks $S$. Hence $S$ is stable in $G^*$. \qed
\end{proof}

Thus there is a linear map that is surjective from the set of stable matchings in $G^*$ to the set of popular max-matchings in $G$
(by Theorem~\ref{thm:stable} and Theorem~\ref{lem:surjective}).
Hence the stable matching polytope of $G^*$ is an extension of the popular 
max-matching polytope of $G$. 

\subsection{A compact extended formulation for the popular max-matching polytope}

The formulation of the stable matching polytope of $G^*$ (from \cite{Rot92}) is given below.
For any node $u$ in $G^*$, the set $\{v' \succ_{u} v\}$ is the set of all neighbors of $u$ in $G^*$ that it prefers to $v$.
Let $\delta^*(u)$ denote the set of edges incident to $u$ in $G^*$.
Let $T = \cup_{a\in A}(\{a_0,\ldots,a_{n_0-2}\}\cup\{d_1(a),\ldots,d_{n_0-1}(a)\})$: every node in $T$ has to be matched in all stable
matchings in $G^*$.

\begin{eqnarray*}
\sum_{w \succ_{a_i} \tilde{b}} x_{(a_i,w)} \ + \ \sum_{z \succ_{\tilde{b}} a_i} x_{(z,\tilde{b})} \ + \ x_{(a_i,\tilde{b})} & \ \ge \ & \ 1 \ \ \ \ \ \ \ \ \ \ \ \forall (a_i,\tilde{b}) \in E^*\\
\ \ \ \ \ \sum_{e\in\delta^*(u)}x_e \  \ \le \  \ 1 \ \ \ \forall u\in A^*\cup B^* \ \ \ \ \ \ \ \ \text{and}\ \ \ \ \ \ \ \ \ x_e \ \ &\ge& \ \ 0 \ \ \ \ \ \ \ \ \ \ \ \ \forall e \in E^*\\
\sum_{e \in \delta^*(u)} x_e   \  \ = \  \ 1 \ \ \ \ \forall \ u \in T \ \ \ \ \ \text{and}\ \ \ \ \ \ x_{(a,b)} \  & = & \ \sum_{i=0}^{n_0-1} x_{(a_i,\tilde{b})} \ \ \forall (a,b) \in E.
\end{eqnarray*}

\begin{itemize}
\item The topmost constraint captures the {\em stability} constraint for edge $(a_i,\tilde{b}) \in E^*$ where $(a,b) \in E$ and 
$0 \le i \le n_0-1$.

\item The constraints in the second line capture that $\vec{x}$ belongs to the matching polytope of $G^*$.

\item The constraints in the third line (on the left) capture
  the stability constraint for the edges $(a_{i-1},d_i(a))$ and  $(a_i,d_i(a))$ for $1 \le i \le n_0-1$.

\item The constraint $x_{(a,b)} = \sum_{i=0}^{n_0-1} x_{(a_i,\tilde{b})}$ for every $(a,b) \in E$ 
  captures the linear map from the edge set $E^*$ to the edge set $E$.
\end{itemize}

We showed that the polytope described by the above constraints is an extension of the popular max-matching polytope of $G$. 
Linear programming on the above formulation with $\min \sum_{e\in E}c(e)\cdot x_e$ as the objective function 
computes a  min-cost popular max-matching in $G$ in polynomial time.

\section{Min-cost Pareto-optimal (max-)matching}
\label{sec:hardness}
In this section we consider the problems of computing a min-cost Pareto-optimal matching and a min-cost Pareto-optimal max-matching 
in a marriage instance with edge costs in $\{0,1\}$. We will show these problems are NP-hard.

Given a 3SAT formula $\psi$, we will build an instance $G_{\psi}$ with edge costs in $\{0,1\}$ such that $G_{\psi}$ admits a Pareto-optimal 
matching of cost 0 if and only if $\psi$ is satisfiable. Any Pareto-optimal matching of cost~0  would have to be a perfect 
matching in $G_{\psi}$. Hence this will prove the NP-hardness of both the min-cost Pareto-optimal matching problem and the min-cost 
Pareto-optimal max-matching problem. 

Our reduction resembles a hardness reduction from \cite{CFKP18} that showed the NP-hardness of deciding if an instance $G$ has a stable 
matching $M$ that is also {\em dominant}, i.e., $M$ is more popular than every larger matching.
As done in this reduction, we will first transform $\psi$ so that every clause contains either only positive literals or only negative 
literals; moreover, there will be a single occurrence of each negative literal in the transformed $\psi$. This is easy to achieve:
\begin{itemize}
\item let $X_1,\ldots,X_n$ be the starting variables. For $i \in [n]$: replace all occurrences of $\neg X_i$ with the same variable
  $X_{n+i}$ (a new one) and add the two clauses $(X_i \vee X_{n+i}) \wedge (\neg X_i \vee \neg X_{n+i})$ to capture $\neg X_i \equiv X_{n+i}$.
  Thus there are $2n$ variables in the transformed $\psi$.
\end{itemize}

We build the graph $G_{\psi}$ as follows. There are two types of gadgets: those that correspond to positive clauses and those that correspond 
to negative clauses. Fig.~\ref{pareto-fig1:example} shows how a positive clause gadget with 3 literals looks like and 
Fig.~\ref{pareto-fig2:example} shows how a negative clause gadget looks like.

\begin{figure}[h]
\centerline{\resizebox{0.36\textwidth}{!}{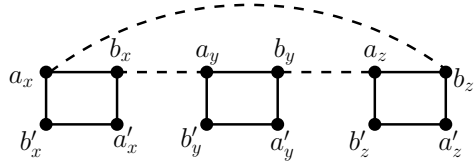}}
\caption{The clause gadget for a positive clause $C_{\ell} = x \vee y \vee z$. Every occurrence of a literal in $\psi$ has a separate gadget. So we ought to use labels such as $a_{x,\ell},b_{x,\ell},\ldots$ here; for the sake of simplicity, we used the labels $a_x,b_x,\ldots$ here.}
\label{pareto-fig1:example}
\end{figure}

We now describe the preference lists of nodes in a positive clause $C_{\ell} = x \vee y \vee z$ (see Fig.~\ref{pareto-fig1:example}). 
The nodes $a_x,a'_x,b_x,b'_x$ occur in $x$'s gadget and $a_y,a'_y,b_y,b'_y$ occur in $y$'s gadget and
$a_z,a'_z,b_z,b'_z$ occur in $z$'s gadget: these gadgets are in the $\ell$-th clause gadget $C_{\ell}$. Every occurrence of a 
literal has a separate gadget. 

\begin{center}
	\begin{tabular}{c|c|c|c|c|c}
	\, $a_x$ \,              & \, $a'_x$ \, & \,   $a_y$ \,                & \, $a'_y$ \, & \, $a_z$ \,                & \, $a'_z$ \,\\ \hline \hline
	\, {$b_z$} \, & \, $b_x$ \,  & \, {$b_x$} \,   & \, $b_y$  \, & \, {$b_y$} \, & \, $b_z$ \, \\ \hline
	\, $b_x$ \,              & \, $b'_x$ \, & \,   $b_y$ \,                & \, $b'_y$ \, & \, $b_z$ \,                 & \, $b'_z$ \, \\ \hline
	\, {\color{red} $d'_x$} \, & \, -- \,     & \, {\color{red} $d'_y$} \,     & \, -- \,     &  \, {\color{red} $d'_z$} \,  & \, -- \,     \\ \hline
        \, $b'_x$ \,               & \, -- \,     & \, $b'_y$ \,                 & \, -- \,    &  \, $b'_z$ \,                  & \, -- \,
	\end{tabular}
	\hspace{1cm}
	\begin{tabular}{c|c|c|c|c|c}
	\, $b_x$ \,              & \, $b'_x$ \, & \,   $b_y$ \,                 & \, $b'_y$ \, & \, $b_z$ \,                & \, $b'_z$ \,\\ \hline \hline
	\, {$a_y$} \, & \, $a'_x$ \,  & \, {$a_z$} \, & \, $a'_y$  \, & \, {$a_x$} \, & \, $a'_z$ \, \\ \hline
	\, $a_x$ \,              & \, {\color{red} $c_x$} \, & \,   $a_y$ \,    & \, {\color{red} $c_y$} \, & \, $a_z$ \,   & \, {\color{red} $c_z$} \, \\ \hline
	\, $a'_x$ \,             & \, $a_x$ \,     & \, $a'_y$ \,               & \, $a_y$ \,     &  \, $a'_z$ \,  & \, $a_z$ \,     
	\end{tabular}
\end{center}

Here $a_x$'s top choice is $b_z$, second choice $b_x$, third choice $d'_x$, fourth choice $b'_x$, and similarly for other nodes.
For every occurrence of a positive literal $x$: there will be a pair of {\em consistency edges} -- the pair $(a_x,d'_x)$ and $(b'_x,c_x)$ 
in Fig.~\ref{pareto-fig3:example} -- between this gadget of $x$ and the unique gadget of $\neg x$. In our preferences,
the neighbors on consistency edges are marked in red.

\begin{figure}[h]
\centerline{\resizebox{0.2\textwidth}{!}{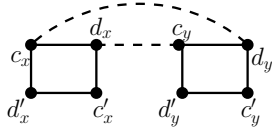}}
\caption{A clause gadget corresponding to a negative clause $D_k = \neg x \vee \neg y$; due to our transformation of $\psi$, every negative clause has only 2 literals.}
\label{pareto-fig2:example}
\end{figure}

The preference lists of nodes that occur in a clause gadget with 2 positive literals will be totally analogous to the
preference lists of nodes in a clause gadget with 3 positive literals.
We will now describe the preference lists of nodes in a negative clause $k$---the overall picture here is given in 
Fig.~\ref{pareto-fig2:example}.

\begin{center}
	\begin{tabular}{c|c|c|c}
	\, $c_x$ \,                  & \, $c'_x$ \, & \,   $c_y$ \,                  & \,  $c'_y$  \,   \\ \hline \hline
	\, {$d_y$} \,     & \, $d_x$ \, & \, {$d_x$} \,     & \,  $d_y$  \,   \\ \hline
	\, $d_x$ \,                  & \, $d'_x$ \,  & \,    $d_y$ \,                 & \, $d'_y$  \,    \\ \hline
        \, {\color{red} $b'_{x,i}$} \,   & \, -- \,     & \, {\color{red} $b'_{y,i'}$} \,   & \,  --  \,   \\ \hline
	\, {\color{red} $\cdots$} \, & \, -- \,     & \, {\color{red} $\cdots$} \,   & \,  --  \,   \\ \hline	
        \, {\color{red} $b'_{x,j}$} \,   & \, -- \,     & \, {\color{red} $b'_{y,j'}$} \,   & \,  --  \,   \\ \hline
        \, $d'_x$ \,                 & \, -- \,     & \, $d'_y$                 \,   & \,  --    \, 
	\end{tabular}
	\hspace{1cm}
        \begin{tabular}{c|c|c|c}
	\, $d_x$ \,                  & \, $d'_x$ \,                 & \,   $d_y$ \,             & \,  $d'_y$  \,   \\ \hline \hline
	\, {$c_y$} \,   & \, $c'_x$ \,                 & \, {$c_x$} \,  & \,  $c'_y$  \,   \\ \hline
	\, $c_x$ \,                  & \, {\color{red} $a_{x,i}$}\,     & \,    $c_y$  \,           & \, {\color{red} $a_{y,i'}$}   \,   \\ \hline
        \, $c'_x$ \,                 & \, {\color{red} $\cdots$} \, & \,    $c'_y$ \,           & \, {\color{red} $\cdots$}  \,   \\ \hline
	\, --     \,                & \, {\color{red}  $a_{x,j}$}\,     & \,     --    \,           & \, {\color{red} $a_{y,j'}$}   \,   \\ \hline	
        \, --     \,                & \, $c_x$ \,                   & \,     --    \,           & \,  $c_y$  \,   
	\end{tabular}
\end{center}

The nodes $c_x,c'_x,d_x,d'_x$ and $c_y,c'_y,d_y,d'_y$ occur in the gadgets of $\neg x$ and $\neg y$, 
respectively. 
The nodes $b'_{x,i},\ldots,b'_{x,j}$ (resp., $b'_{y,i'},\ldots,b'_{y,j'}$) in the preference lists above are the $b'$-nodes in
the $x$-gadgets (resp., $y$-gadgets) in the various clauses that $x$ (resp., $y$) occurs in. Similarly, $a_{x,i},\ldots,a_{x,j}$ 
(resp., $a_{y,i'},\ldots,a_{y,j'}$) 
are the $a$-nodes in the $x$-gadgets (resp., $y$-gadgets) in the various clauses that $x$ (resp., $y$) occurs in. 
The preference order among the $b'$-nodes and among the $a$-nodes in these lists is not important.
The consistency edges between a gadget of $x$ and the gadget of $\neg x$ are shown in Fig.~\ref{pareto-fig3:example}.

\begin{figure}[h]
\centerline{\resizebox{0.67\textwidth}{!}{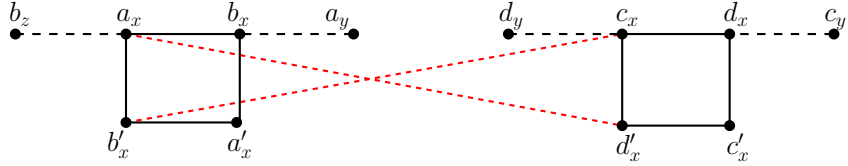}}
\caption{For the sake of simplicity, we use $a_x,b_x,a'_x,b'_x$ to denote the 4 nodes in the gadget of $x$ in the $\ell$-th clause; $c_x,d_x,c'_x,d'_x$ are the 4 nodes in the unique gadget of $\neg x$. The consistency edges are the red dashed edges.} 
\label{pareto-fig3:example}
\end{figure}

\paragraph{\bf Edge costs.} For each edge $e$ in $G_{\psi}$, we will set $\cost(e) \in \{0,1\}$ as follows.
\begin{itemize}
\item For each variable $r \in \{X_1,\ldots,X_{2n}\}$: set $\cost(e) = 0$ where $e$ is any of the 4 edges in any literal gadget 
$\langle a_r,b_r,a'_r,b'_r\rangle$ of $r$ or any of the 4 edges in the gadget $\langle c_r,d_r,c'_r,d'_r\rangle$ of $\neg r$. 
\item For all other edges $e$, set $\cost(e) = 1$. 
\end{itemize}

In particular, for any edge $e$ in the consistency pair for any variable, we have $\cost(e) = 1$. 
In our figures, all dashed edges have cost~1 and all solid edges have cost 0.
Let $M$ be a Pareto-optimal matching in $G_{\psi}$ with $\cost(M) =  0$. So $M$ has to use only cost~0 edges.
Thus $M$ is forbidden to use any edge other than the 4 edges in the gadget of any 
literal. Moreover, since $M$ is Pareto-optimal, $M$ cannot leave two adjacent nodes unmatched. Thus for $r \in \{X_1,\ldots,X_{2n}\}$:

\begin{enumerate}
\item From a gadget of $r$ (say, on nodes $a_r,b_r,a'_r,b'_r$), either (i)~$(a_r,b_r),(a'_r,b'_r)$ are in $M$ or (ii)~$(a_r,b'_r)$, $(a'_r,b_r)$ are in $M$.
\item From the gadget of $\neg r$ (the nodes are $c_r,d_r,c'_r,d'_r$), either (i)~$(c_r,d_r),(c'_r,d'_r)$ are in $M$ or (ii)~$(c_r,d'_r)$, $(c'_r,d_r)$ are in $M$.
\end{enumerate}

Thus any Pareto-optimal matching in $G_{\psi}$ of cost 0 is a perfect matching. It is easy to show Lemma~\ref{pareto:consistency}. 

\begin{lemma}
\label{pareto:consistency}
Let $M$ be a Pareto-optimal matching in $G_{\psi}$. For any $r \in \{X_1,\ldots,X_{2n}\}$,
both $(a_r,b'_r)$ and $(c_r,d'_r)$ cannot simultaneously be in $M$.
\end{lemma}
\begin{proof}
The preferences of the nodes are set such that if both $(a_r,b'_r)$ and $(c_r,d'_r)$ are in $M$ then both the non-matching edges 
$(a_r,d'_r)$ and $(b'_r,c_r)$ in the alternating cycle $\rho = a_r - d'_r - c_r - b'_r - a_r$ are blocking edges to $M$. 
Consider $M \oplus \rho$ versus $M$. All the 4 nodes $a_r, b'_r, c_r, d'_r$ prefer $M\oplus\rho$ to $M$ while the other nodes are
indifferent between $M \oplus \rho$ and $M$. 
Thus  $\phi(M\oplus\rho, M) = 4$ and $\phi(M, M\oplus\rho) = 0$, so $u(M) = \infty$.
This means $M$ is not Pareto-optimal, a contradiction.
Thus for any $r \in \{X_1,\ldots,X_{2n}\}$, we cannot have both $(a_r,b'_r)$ and $(c_r,d'_r)$ in $M$. \qed
\end{proof}

Theorem~\ref{thm:pareto-correctness} is our main result in this section. 

\begin{theorem}
\label{thm:pareto-correctness}
$G_{\psi}$ has a Pareto-optimal matching $M$ with $\cost(M) = 0$ if and only if $\psi$ is satisfiable. 
\end{theorem}
\begin{proof}
Suppose $G_{\psi}$ has a Pareto-optimal matching $M$ with $\cost(M) = 0$.
Lemma~\ref{pareto:consistency} indicates a natural way of defining an assignment for $\psi$.
For any variable $r \in \{X_1,\ldots,X_{2n}\}$, consider the edges in $\neg r$'s gadget that are in $M$. 
If $(c_r,d'_r),(c'_r,d_r)$ are in $M$ then set $r = \mathsf{false}$ else set $r = \mathsf{true}$. 

Lemma~\ref{pareto:consistency} tells us that when $r = \mathsf{false}$, the edges $(a_{r,i},b_{r,i}),(a'_{r,i},b'_{r,i})$ 
from $r$'s gadget in the $i$-th clause are in $M$ (where $a_{r,i},b_{r,i},a'_{r,i},b'_{r,i}$ are the 4 nodes from $r$'s gadget in the $i$-th
clause). 

\begin{new-claim}
  \label{clm:app1}
The above assignment satisfies $\psi$.
\end{new-claim}

Claim~\ref{clm:app1} uses the Pareto-optimality of $M$ to show that every clause has at least one literal set to $\mathsf{true}$. Its proof is given after the proof of Theorem~\ref{thm:pareto-correctness}. Hence if $G_{\psi}$ admits a Pareto-optimal matching $M$ with $\cost(M) = 0$, then $\psi$ is satisfiable.

\smallskip

\noindent{\em The converse.}
We will now show that if $\psi$ is satisfiable then there is a Pareto-optimal matching $M$ in $G_{\psi}$ such that $\cost(M) = 0$.
There is a natural way of constructing the matching $M$---we will use the satisfying assignment for $\psi$ to choose edges from each 
literal gadget. For any variable $r$, include the following edges in the matching $M$:
\begin{itemize}
\item if $r = \mathsf{true}$ then take the pair of edges $(c_r,d_r),(c'_r,d'_r)$ from $\neg r$'s gadget and the pair of edges 
$(a_{r,i},b'_{r,i})$, $(a'_{r,i},b_{r,i})$ from $r$'s gadget in clause~$i$ (for every clause $i$ that $r$ belongs to). 
\item if $r = \mathsf{false}$ then take the pair of edges $(c_r,d'_r),(c'_r,d_r)$ from $\neg r$'s gadget and the pair of edges 
$(a_{r,i},b_{r,i})$, $(a'_{r,i},b'_{r,i})$ from $r$'s gadget in clause~$i$ (for every clause $i$ that $r$ belongs to). 
\end{itemize}

It is easy to see that $\cost(M) = 0$. Since $M$ is a perfect matching, there is no alternating path $\rho$ wrt $M$ such that 
$\phi(M, M\oplus\rho) = 0$. This is because for every alternating path $\rho$ wrt $M$, we have
$|M\oplus\rho| < |M|$ and the nodes matched in $M$ and unmatched in $M\oplus\rho$ prefer $M$ to $M\oplus\rho$, so
$\phi(M, M\oplus\rho) > 0$.
Hence in order to prove $M$'s Pareto-optimality, what we need to show is Claim~\ref{clm:app2}.

\begin{new-claim}
  \label{clm:app2}
There is no alternating cycle $\rho$ with respect to $M$ such that $\phi(M\oplus\rho,M) > 0$ and $\phi(M,M\oplus\rho) = 0$.
\end{new-claim}

The proof of Claim~\ref{clm:app2} is given below. This finishes the proof of Theorem~\ref{thm:pareto-correctness}. \qed
\end{proof}

\paragraph{\bf Proof of Claim~\ref{clm:app1}.}
Suppose this assignment does not satisfy $\psi$. We have 3 cases here.
\begin{enumerate}
\item Let $C_i = x \vee y \vee z$. Suppose all the three variables $x, y, z$ are in false state. Consider the following alternating 
cycle $\rho$ wrt $M$: 
\[b_{z,i} - (a_{x,i},b_{x,i}) - (a_{y,i},b_{y,i}) - (a_{z,i},b_{z,i}) - a_{x,i}.\] 

All non-matching edges in this alternating cycle, i.e., the edges $(b_{z,i},a_{x,i})$,  $(b_{x,i},a_{y,i})$, $(b_{y,i},a_{z,i})$,
are blocking edges with respect to $M$. In the $M \oplus \rho$ versus $M$ comparison, these 6 nodes 
$a_{x,i}, b_{x,i}, a_{y,i}, b_{y,i}, a_{z,i}, b_{z,i}$ prefer $M \oplus \rho$ to $M$ while all other nodes in $G_{\psi}$ are indifferent 
between $M \oplus \rho$ and $M$. 
Thus we have $\phi(M\oplus\rho, M) = 6$ and $\phi(M, M\oplus\rho) = 0$. Hence $u(M) = \infty$, contradicting the Pareto-optimality of $M$.

\smallskip

\item Let $C_j = x \vee y$, i.e., this is a positive clause with 2 literals. Suppose both $x$ and $y$ are in false state.
Consider the following alternating cycle $\rho$ wrt $M$:
\[ b_{y,j} - (a_{x,j},b_{x,j}) - (a_{y,j},b_{y,j}) - a_{x,j}.\] 

In the $M \oplus \rho$ versus $M$ comparison, the 4 nodes $a_{x,j}, b_{x,j}, a_{y,j}, b_{y,j}$ prefer $M \oplus \rho$ to $M$ 
while all the other nodes in $G_{\psi}$ are indifferent between $M\oplus\rho$ and $M$. Thus $\phi(M\oplus\rho, M) = 4$ and
$\phi(M, M\oplus\rho)=0$. Hence $u(M) = \infty$, contradicting the Pareto-optimality of $M$.

\smallskip

\item Let $D_k = \neg x \vee \neg y$. Suppose both $\neg x$ and $\neg y$ are in false state.
Consider the following alternating cycle $\rho$ wrt $M$: 
\[d_y - (c_x,d_x) - (c_y,d_y) - c_x.\] 

In the $M \oplus \rho$ versus $M$ comparison, the 4 nodes $c_x, d_x, c_y, d_y$ prefer $M \oplus \rho$ to $M$ 
while all the other nodes in $G_{\psi}$ are indifferent between $M\oplus\rho$ and $M$. So $\phi(M\oplus\rho, M) = 4$ and 
$\phi(M, M\oplus\rho)=0$. Hence $u(M) = \infty$, contradicting the Pareto-optimality of $M$. 
\end{enumerate}

Thus every clause in $\psi$ has at least one literal in true state.  \qed

\paragraph{\bf Proof of Claim~\ref{clm:app2}.}
We need to show there is no alternating cycle $\rho$ wrt $M$ such that $\phi(M\oplus\rho,M) > 0$ and $\phi(M,M\oplus\rho) = 0$.  
Every non-matching edge in such an alternating cycle $\rho$ has to be a blocking edge wrt $M$. 
First, we argue that every consistency edge is a {\em non-blocking} edge to $M$; say, this is a consistency edge corresponding 
to variable $r$ in clause~$i$. It follows from our construction of $M$ that $M$ contains either:
\begin{enumerate} 
\item $(a_{r,i},b'_{r,i}), (a'_{r,i},b_{r,i})$ and $(c_r,d_r), (c'_r,d'_r)$ or 
\item $(a_{r,i},b_{r,i}), (a'_{r,i},b'_{r,i})$ and $(c_r,d'_r), (c'_r,d_r)$.
\end{enumerate}

\begin{itemize}
\item In case~1: the node $d'_r$ prefers $M(d'_r) = c'_r$ to $a_{r,i}$ and the node $c_r$ prefers $M(c_r) = d_r$ to $b'_{r,i}$. 
\item In case~2: the node $a_{r,i}$ prefers $M(a_{r,i}) = b_{r,i}$ to $d'_r$ and the node $b'_{r,i}$ prefers $M(b'_{r,i}) = a'_{r,i}$ to $c_r$. 
\end{itemize}
Thus in both cases, the consistency edges $(a_{r,i},d'_r)$ and $(b'_{r,i},c_r)$ are non-blocking edges to $M$. Let $H$ be the subgraph 
of $G_{\psi}$ obtained by preserving only edges in $M$ and blocking edges wrt $M$. Thus no non-blocking edge (other than edges in $M$)
is included in $H$ -- so no consistency edge belongs to $H$.

Since there are no consistency edges in $H$, any alternating cycle in $H$ has to be contained within a single clause. 
We will now show there is no such cycle in $H$ by using the fact that we constructed $M$ using a satisfying assignment for $\psi$:
thus every clause has at least one literal set to $\mathsf{true}$. 

Let $C = x \vee y \vee z$ and suppose $y = \mathsf{true}$ in $\psi$. Then $(a_y,b'_y)$ and $(a'_y,b_y)$ are in $M$, however the edge  
$(a'_y,b'_y)$ is non-blocking wrt $M$ and hence it is missing in $H$. Thus there is no alternating cycle in $H$ that is contained
within the clause $C$. Now consider a negative clause $D = \neg x \vee \neg y$ and suppose $x = \mathsf{false}$ in $\psi$.
Then $(c_x,d'_x)$ and 
$(c'_x,d_x)$ are in $M$, however the edge  $(c'_x,d'_x)$ is non-blocking wrt $M$ and it is missing in $H$. Thus there is no alternating 
cycle in $H$ that is contained within the clause $D$.

Consider the 4 edges of any literal gadget (say, $r$) in $G_{\psi}$: if $(a_r,b_r) \in M$ then 
$(a_r,b'_r)$ is a non-blocking edge wrt $M$ and if $(a'_r,b_r) \in M$ then $(a'_r,b'_r)$ is a non-blocking edge wrt $M$. 
Similarly, in the gadget of $\neg r$: if $(c_r,d_r) \in M$ then $(c_r,d'_r)$ is a non-blocking edge wrt $M$ 
and if $(c'_r,d_r) \in M$ then $(c'_r,d'_r)$ is a non-blocking edge wrt $M$.
Thus there is no alternating cycle wrt $M$ in $H$.
So there is no alternating cycle $\rho$ in $G_{\psi}$ such that $\phi(M\oplus\rho, M) > 0$ and $\phi(M, M\oplus\rho) = 0$. \qed

\medskip

Observe that any Pareto-optimal matching in $G_{\psi}$ of cost 0 is a perfect
matching. Thus Theorem~\ref{thm:pareto-correctness} shows that the min-cost Pareto-optimal matching problem 
and the min-cost Pareto-optimal max-matching problem are NP-hard.
Moreover, these problems are NP-hard to approximate to any multiplicative factor. Thus we have shown the following theorem.

\begin{theorem}
\label{thm:pareto-opt}
Given $G = (A \cup B, E)$ with strict preferences and edge costs in $\{0,1\}$,
it is NP-hard to compute a min-cost Pareto-optimal matching in $G$;
moreover, it is NP-hard to approximate this within any multiplicative factor.
\end{theorem}

\smallskip

\noindent{\em Acknowledgments.} Thanks to Yuri Faenza for helpful comments that improved the presentation.

\newpage

\section*{Appendix}
\paragraph{\bf Missing details from the proof of Theorem~\ref{thm:stable}.}
The proofs of properties~1-6 are given below.

\paragraph{Proof of property~1.}
The inclusion $S' \subseteq \cup_{i=0}^{n_0-1} (A_i\times B_i)$ follows from the definition of the sets $B_i$. Recall that for 
$1 \le i \le n_0-1$, $B_i$ is the set of nodes $b$ such that $(a_i,\tilde{b}) \in S$ for some $a \in A_i$. Also, $B_0$ 
contains all nodes $b$ such that $(a_0,\tilde{b}) \in S$ for some $a \in A_0$. Thus $S' \subseteq \cup_{i=0}^{n_0-1} (A_i\times B_i)$.

     The stability of $S'$ restricted to each set $A_i \cup B_i$
     is by $\tilde{b}$'s preference order in $G^*$. Recall that within subscript~$i$ neighbors, the order of preference for $\tilde{b}$ 
     in $G^*$ is $b$'s
     order of preference in $G$. Thus the stability of $S$ in $G^*$ implies the stability of $S'$ restricted to $A_i \cup B_i$ for each $i$.

\paragraph{Proof of property~2.}
Let $a \in A_{i+1}$. Then $(a_{i+1},d_{i+1}(a)) \notin S$. So it has to be the case that 
     $(a_i,d_{i+1}(a)) \in S$. Recall that $d_{i+1}(a)$ is $a_i$'s least preferred neighbor in $G^*$. So $a_i$ prefers $\tilde{b}$ to its 
     partner in $S$. Hence it follows from the stability of $S$ in $G^*$ that $\tilde{b}$ prefers its partner in $S$ (this is a 
     subscript~$i$ node $z_i$) to $a_i$, i.e., $b$ prefers $z$ to $a$.

       Since $\tilde{b}$ prefers subscript~$i+1$ nodes to subscript~$i$ nodes, $\tilde{b}$ prefers $a_{i+1}$ to its partner $z_i$ 
       in $S$. It follows from the stability of $S$ in $G^*$ that $a_{i+1}$ has to prefer its partner $\tilde{w}$ in $S$ to $\tilde{b}$,
       otherwise $(a_{i+1},\tilde{b})$ would block $S$. Hence $a$ prefers $w$ to $b$. Thus $\wt_{S'}(a,b) = -2$.

\paragraph{Proof of property~3.}
     Suppose $a \in A_i$ where $i \ge j+2$ and $b \in B_j$. So the edge $(a_{j+1}, d_{j+2}(a)) \in S$.
     Since $d_{j+2}(a)$ is $a_{j+1}$'s least preferred neighbor in $G^*$,
     the stability of $S$ implies that $\tilde{b}$ prefers its partner in $S$ to $a_{j+1}$.
     However $b \in B_j$ and so $\tilde{b}$'s partner in $S$ is a subscript~$j$ node $z_j$. This contradicts $\tilde{b}$'s 
     preference order that it prefers any subscript~$j+1$ neighbor to a subscript~$j$ neighbor.
     Thus there is no edge $(a,b)$ in $G$ with $a \in A_i$ and $b \in B_j$ where $i \ge j+2$.

\paragraph{Proof of property~4.}
     This follows from properties 1, 2, and 3 given above. 
     Properties~2 and 3 tell us that there is no blocking edge in $A_i \times B_j$ where $i \ge j+1$.
     Property~1 tells us that there is no blocking edge in $A_i \times B_i$ for any $i$.
     So any blocking edge to $S'$ has to be in $A_i \times B_j$ where $i \le j-1$.

\paragraph{Proof of property~5.}
         This follows from the definitions of the sets $A_0,\ldots,A_{n_0-2}$ and $B_1,\ldots,B_{n_0-1}$. For each $a \in A_i$ where 
         $0 \le i \le n_0-2$: we have $(a_i,\tilde{b}) \in S$ for some $\tilde{b} \in \tilde{B}$ and thus $(a,b) \in S'$. 
         Similarly, for each $b \in B_j$ where $1 \le j \le n_0-1$: we have $(a_j,\tilde{b}) \in S$ for some $a \in A_j$ and 
         thus $(a,b) \in S'$. Hence all nodes that are unmatched in $S'$ are in $A_{n_0-1} \cup B_0$.

\paragraph{Proof of property~6.}
  Suppose $S'$ is not a maximum matching in $G$. Then there is an augmenting path $\rho$ with respect to $S'$. Let us refer to an edge $e$ 
  that satisfies $\wt_{S'}(e) = -2$ as a {\em negative} edge. The endpoints of a negative edge prefer their respective assignments in $S'$
  to each other. We know from  property~5 above that all nodes in $A$ that are unmatched in $S'$ are in $A_{n_0-1}$ and all nodes
  in $B$ that are unmatched in $S'$ are in $B_0$. We also know that $S' \subseteq \cup_i (A_i \times B_i)$ (by prop.~1 above).
  Moreover, all edges $e$ in $A_{j+1} \times B_j$ are negative edges (by prop.~2 above) and there is no edge in $A_i \times B_j$
  where $i \ge j+2$ (by prop.~3 above).

  Thus the path $\rho$ starts in $A_{n_0-1}$ at an unmatched node $a$ and since there cannot be any negative edge incident to an unmatched
  node, all of $a$'s neighbors have to be in $B_{n_0-1}$: this is because every edge $e$ in $A_{n_0-1}\times B_{n_0-2}$ is a negative edge.
  The matched partners of $a$'s neighbors are in $A_{n_0-1}$. Then the next node can be in $B_{n_0-2}$ (this is by prop.~3 above) and its partner
  is in $A_{n_0-2}$ and so on. Finally, there is no edge from $A_1$ to an unmatched node in $B_0$: this is because there is no negative edge
  incident to an unmatched node and we know all edges in $A_1 \times B_0$ are negative edges (by prop.~2 above).

  So the {\em shortest} alternating path $\rho$ from an unmatched $a \in A_{n_0-1}$ to an unmatched $b \in B_0$ moves across sets as follows:
  $A_{n_0-1} - B_{n_0-1} - A_{n_0-1} - B_{n_0-2} - A_{n_0-2} - B_{n_0-3} - \cdots - A_1 - B_0 - A_0 - B_0$.
  This implies there are at least $n_0+1$ nodes in $A$. However $|A| = n_0$. So there is no such alternating path, i.e., 
  there is no augmenting path with respect to $S'$. In other words, $S'$ is a maximum matching in $G$. \qed

  \paragraph{\bf Proof of Claim~\ref{claim2}.}
    Let $v_0$ be the first node that moves to $A_{n_0-1}$ with a blocking edge incident to it. Recall our update procedure --
    we initially added nodes in $\Nbr(U_A)$ and their partners in $M$ to $B_{n_0-1}$ and $A_{n_0-1}$, respectively.
    Then we applied rules~1-3 in some order and this resulted in the node $v_0$ moving to $A_{n_0-1}$.
    Corresponding to these rules, we will construct an alternating path
    $p = v_0 - M(v_0) - v_1 - M(v_1) - v_2 - \cdots - v_k - M(v_k) - u$  between $v_0$ and some node $u \in U_A$.
    A useful observation is that all the nodes $v_0,\ldots,v_k$ have moved upwards to their respective levels due to our
    update of $\vec{\alpha}$ to ensure property~(ii) and to maintain $\vec{\alpha}$ as an optimal solution to \eqref{LP4}.
    
    The path $p$ can be partitioned into $k+1$ pairs of edges for some $k \ge 0$. For $0 \le i \le k-1$: the $i$-th pair consists of the
    matching edge $(v_i,M(v_i))$ of weight 0 and the non-matching edge $e_i = (M(v_i),v_{i+1})$
    where $v_{i+1}$ is the node that pulled $M(v_i)$ and $v_i$ to their current level due to the application of one
    of the above three rules. Rule~1 implies $\wt_M(e_i) = 0$ while rule~2 implies $\wt_M(e_i) = -2$ and rule~3 implies $\wt_M(e_i) = 2$.

    Observe that rule~1 places $M(v_i)$ in the same level as $v_{i+1}$ while rule~2 places $M(v_i)$ one level
    lower than $v_{i+1}$ and rule~3 places $M(v_i)$ one level higher than $v_{i+1}$.
    The last pair of edges in $p$ are $(v_k, M(v_k)) \in A_{n_0-1} \times B_{n_0-1}$ and $(M(v_k),u)$, where the latter edge has weight $0$. 
    The node $v_0$ is in level~$n_0-1$ and the node $M(v_k)$ is also in level $n_0-1$.
    So $v_0$ and $v_k$ are at the same level, hence the number of edges in $p$ of weight $-2$ is exactly the same as the number of edges
    of weight $2$, thus $\wt_M(p) = 0$.

    Suppose there is a blocking edge $(v_0,w)$. If $w$ belongs to $p$, then it is easy to see that the alternating
    cycle $C$ obtained by joining the endpoints of the $v_0$-$w$ subpath in $p$ with the edge $(v_0,w)$ satisfies $\wt_M(C) \ge 2$.
    This contradicts Theorem~\ref{thm:char} since
    $M$ is a popular max-matching. Hence $w$ does not belong to path $p$.
    So let us add the 2-edge path $M(w) - w - v_0$ as a prefix to the $v_0$-$u$ path $p$ and call this alternating path $q$: we have
    $\wt_M(q) = \wt_M(p) + \wt_M(v_0,w) = 2$. 
    Since $\wt_M(q) > 0$ and the unmatched node $u$ is an endpoint of $q$, this again contradicts Theorem~\ref{thm:char}.
    Hence there is no blocking edge incident to $v_0$. \qed
\end{document}